\documentclass[conference]{IEEEtran}
\usepackage{flushend}
\usepackage{cite}
\usepackage{amsmath,amssymb}    
\usepackage[dvips]{graphicx}   
\usepackage{verbatim}   
\usepackage{color}      
\usepackage{subfigure}  
\usepackage{booktabs}
\usepackage{multirow}
\usepackage{graphics}
\usepackage{epsfig}
\usepackage{float}
\usepackage{algorithm}
\usepackage{algorithmic}
\usepackage{subfigure}
\newtheorem{theorem}{\bf Theorem}

\newenvironment{proof}[1][Proof]{\begin{trivlist}
\item[\hskip \labelsep {\bfseries #1}]}{\end{trivlist}}

\makeatletter
\newcommand*{\rom}[1]{\expandafter\@slowromancap\romannumeral #1@}

\newcommand{\newtext}[1]{{\color{black}{#1}}}
\newcommand{\IntS}{\newtext{I_s}}

\makeatother
\IEEEoverridecommandlockouts
\begin{document}
\title{\textbf{Byzantine-Resilient Locally Optimum Detection Using Collaborative Autonomous Networks}\vspace{-0.4in}}
\author{\IEEEauthorblockN{Bhavya Kailkhura, Priyadip Ray, Deepak Rajan, Anton Yen, Peter Barnes, Ryan Goldhahn}
\IEEEauthorblockA{Lawrence Livermore National Laboratory}
\thanks{This work was performed under the auspices of the U.S. Department of Energy by Lawrence Livermore National Laboratory under Contract DE-AC52-07NA27344. LLNL-CONF-731964
}
}
\maketitle
\begin{abstract}
In this paper, we propose a locally optimum detection (LOD) scheme for detecting a weak radioactive source buried in background clutter.
We develop a decentralized algorithm, based on alternating direction method of multipliers (ADMM), for implementing the proposed scheme in autonomous sensor networks.
Results show that algorithm performance approaches the centralized clairvoyant detection algorithm in the low SNR regime, and 
exhibits excellent convergence rate and scaling behavior (w.r.t. number of nodes). We also devise a low-overhead, robust ADMM algorithm for Byzantine-resilient detection, and demonstrate its robustness to data falsification attacks.
\end{abstract}


\section{Introduction}
\label{intro}
Autonomous vehicles provide sensing platforms which are small, low cost, and maneuverable. Because of size, weight and power restrictions, the sensors onboard are of limited performance. Signal processing and data fusion techniques are thus needed to approach the performance of a more capable sensor with a large number of adaptively re-configurable low cost sensors.  This work provides a computationally tractable scheme for autonomous detection, applied to the problem of detecting a radioactive source.

Detection of radioactive sources using sensor networks has received significant attention in the literature. In~\cite{nd1}, the authors examine the gain in signal-to-noise ratio obtained by a simple combination of data from networked sensors compared to a single sensor for radioactive source detection. The costs and benefits of using a network of radiation detectors for radioactive source detection are analyzed and evaluated in~\cite{nd2}. In~\cite{ashok_fusion}, the authors derived a test for the fusion of correlated decisions and obtained optimal sensor thresholds for two sensor case. In~\cite{nd3}, the authors considered the problem of detecting a time-inhomogeneous Poisson process buried in the recorded background radiation using sensor networks. However, all these works assume existence of a centralized fusion center (FC) to fuse the data from multiple sensors and to make a global decision.

In many scenarios, a centralized FC may not be available. 
Furthermore, due to the distributed nature of future communication networks and various practical constraints (e.g., absence of the FC, transmit power or hardware constraints, and dynamic characteristic of wireless communications), it may be desirable to achieve collaborative decision making by employing peer-to-peer local information exchange to reach a global decision. 
Recently, collaborative autonomous detection based on consensus algorithms has been explored in~\cite{paper4,paper5,paper6,paper7,paper8,paper9}. 
However, all these approaches assume a clairvoyant detection where all the parameters of the detection system and signal model are completely known. Note that, for our application of interest (i.e., nuclear radiation detection) the location of the radiation sources is rarely known. Centralized approaches manage this challenge by employing composite hypothesis testing frameworks such as the generalized likelihood ratio test (GLRT). In GLRT, the detection procedure replaces unknown parameters in the detection algorithm with their maximum likelihood estimates, which need multiple sensing intervals for a reasonably accurate parameter estimate. This overhead and delay is not desirable in nuclear radiation detection problems, especially under weak signal models. Secondly, due to the non-linearity introduced by the estimation step in GLRT, a decentralized implementation of GLRT is non-trivial. 
Finally, the implementation of non-linear detectors on low cost UAVs is difficult in practice. Thus, a decentralized solution with a simple implementation for the radiation detection problem with unknown source location is of utmost interest.

Autonomous detection schemes are quite vulnerable to different types of attacks. One typical attack on such networks is a Byzantine (or data falsification) attack~\cite{frag, Rifa, Marano, Rawat, bhavyaj, Kailkhura, aditya}. 
Few attempts have been made to address the Byzantine attacks in conventional consensus-based detection schemes in recent research~\cite{bk_consensus,tang,yu1,yu2,liu,yan,paper12}.
There exist several methods for decentralized consensus optimization, including distributed subgradient descent algorithms~\cite{nedic}, dual averaging methods~\cite{duchi}, and the alternating direction method of multipliers (ADMM)~\cite{boyd_admm}. Among these, the ADMM has drawn significant attention, as it is well suited for distributed convex optimization and demonstrates fast convergence in many applications. However, the performance analysis of ADMM in the presence of data falsifying Byzantine attacks has thus far not been addressed in the literature.

To overcome the aforementioned challenges, in this paper we propose a simple to implement locally optimum detection algorithm to detect radioactive source signal buried in noise. We also devise a robust variant of ADMM algorithm to implement this detection scheme in autonomous networks in the presence of Byzantine attacks. To the best of our knowledge, there have been no existing results on the Byzantine-resilient locally optimum detection in collaborative autonomous sensor networks.

\section{System Model}


\subsection{Signal Model}
\label{sigmodel}
Consider two hypotheses $H_0$ (radioactive source is absent) and $H_1$ (radioactive source is present). Also, consider a network of $N$ autonomous {nodes which must determine which} of the two
hypotheses is true. The observations received by
the node $i$ for $i=1,\cdots,N$ under both hypotheses are as follows.
\begin{eqnarray}
H_0 &:& z_i = b_i+w_i\nonumber \\
H_1 &:& z_i = c_i+b_i+w_i 
\end{eqnarray}
where \newtext{$b_i$,} $c_i$ and $w_i$ are the background radiation count,
source radiation count and measurement noise respectively, at
node $i$ located at $\{X_i, Y_i\}$\footnote{Note that, the proposed scheme can easily be extended to a three-dimensional setting}. The
background radiation count is assumed to be Poisson distributed with known rate parameter $\lambda_b$. The source radiation count at node $i$ is assumed to be Poisson distributed with rate parameter $\lambda_{ci}$. We assume
an isotropic behavior of radiation in the presence of the \newtext{source;
the} rate $\lambda_{ci}$ is a function of the source intensity \newtext{$\IntS$ and} distance of the $i$th sensor from the \newtext{source, given} by
\begin{equation}
\lambda_{ci} = \dfrac{\IntS}{(X_i-X_s)^2+(Y_i-Y_s)^2},
\end{equation}
\noindent where $\{X_s, Y_s\}$ represent the source coordinates. The measurement
noise $w_i$ is Gaussian distributed with a known variance
$\sigma_w^2$. The background radiation count $b_i$ and measurement
noise $w_i$ are assumed to be independent. We also assume \newtext{that the  observations at any node are} conditionally independent and identically
distributed given the hypothesis. It is well known that the above signal model can be approximated by the Gaussian distribution~\cite{ashok_fusion}. Thus, under $H_0$, we have
\newtext{$$f_0(z_i) = \mathcal{N}(\lambda_b, \lambda_b+\sigma_w^2).$$}
Similarly, \newtext{under the $H_1$} hypothesis,
\newtext{$$f_1(z_i) = \mathcal{N}(\lambda_{ci}+\lambda_b, \lambda_{ci}+\lambda_b+\sigma_w^2),$$}
where $\lambda_{ci}$ is a function of node $i$'s position relative to 
source.

\subsection{Collaborative Autonomous Detection: Clairvoyant Case}
For ease of exposition, we first consider the clairvoyant case, i.e., the values of source intensity $\IntS$ and source coordinates $\{x_s, y_s\}$
are assumed to be known. In our setting, however, the source location is unknown, which is addressed in detail in subsequent sections.
The collaborative autonomous detection scheme usually contains three phases: $1)$ sensing, $2)$ collaboration, and $3)$ decision making. 
In the sensing phase, each node acquires the summary statistic about the phenomenon of interest. Next, in the collaboration phase, each node communicates with its neighbors to update/improve their state values (summary statistic) and continues with this process until the whole network converges to a steady state which is the global test statistic.
Finally, in the decision making phase, nodes make their own decisions about the presence of the phenomenon using this global test statistic. 

The clairvoyant detector
is easy to implement in a decentrlized setup using a consensus based approaches~\cite{paper2} and is the log likelihood ratio test (LRT) given by:
$$\sum\limits_{i=1}^{N}\log\left(\dfrac{f_1(z_i)}{f_0(z_i)}\right) \quad \mathop{\stackrel{H_1}{\gtrless}}_{H_0} \quad  \log\lambda,$$
where $\lambda$ is chosen such \newtext{that} the probability of false alarm is constrained below a pre-specified level $\delta$.

\subsection{Detection with Unknown Source Location: GLRT}
In many practical scenarios, \newtext{including the focus of this work,} the location of the radioactive source is not known and the LRT cannot be implemented. In such scenarios, one of the most popular
\newtext{tests} is the Generalized Likelihood Ratio Testing (GLRT). The GLRT has an estimation procedure built into it, where the underlying parameter estimates are used as a plug-in estimate for the test statistic. More specifically, the GLRT test statistic is as follows:
\begin{equation}
\max_{\lambda_{ci}} \sum\limits_{i=1}^{N}\log\left(\dfrac{f_1(z_i;\lambda_{ci})}{f_0(z_i)}\right) \quad \mathop{\stackrel{H_1}{\gtrless}}_{H_0} \quad  \log\lambda.
\end{equation}
\newtext{Next,} we show that in the low signal to noise ratio (SNR) regime, there exist a locally optimum detection scheme which alleviates the difficulties (e.g., delay, overhead and non-linearity) in implementing GLRT in an autonomous setting. 

\section{Collaborative Autonomous Locally Optimum Detection (CA-LOD)}
For \newtext{ease} of exposition, we first \newtext{derive the new locally optimum detection scheme for a centralized scenario.} Then, we \newtext{present}
an approach to implement the proposed detection scheme in a decentralized \newtext{setting.} 

\subsection{Locally Optimum Centralized Detection}
\begin{theorem}
The locally optimal test statistic is \newtext{given} by
\begin{equation}
\label{lodstat}
\sum\limits_{i=1}^{N} (z_i-\lambda_b)+\sum\limits_{i=1}^{N} \dfrac{(z_i-\lambda_b)^2}{2(\lambda_b+\sigma_w^2)}\quad \mathop{\stackrel{H_1}{\gtrless}}_{H_0} \quad  \gamma,
\end{equation}
where $\gamma$ is chosen such that the probability of false alarm is constrained below a pre-specified level $\delta$.
\end{theorem}
\begin{proof}
The LRT for known $\lambda_{ci}$ is given by
\begin{small}
\begin{eqnarray}
&&\sum\limits_{i=1}^{N}\log\left(\dfrac{f_1(z_i;\lambda_{ci})}{f_0(z_i)}\right) \quad \mathop{\stackrel{H_1}{\gtrless}}_{H_0} \quad  \log\lambda\\
&\Leftrightarrow& \sum\limits_{i=1}^{N}\log f_1(z_i;\lambda_{ci}) - \sum\limits_{i=1}^{N}\log f_0(z_i)\quad \mathop{\stackrel{H_1}{\gtrless}}_{H_0} \quad  \log\lambda \label{ts}
\end{eqnarray}
\end{small}
However, since we are considering a weak signal scenario, $\lambda_{ci}$ tends to zero, and hence linearizing the LRT around \newtext{$\lambda_{ci} = 0$ results} in,
\begin{small}
\begin{eqnarray*}
&&\sum\limits_{i=1}^{N}(\lambda_{ci}-0)\dfrac{d}{d\lambda_{ci}} \log f_1(z_i;\lambda_{ci})|_{\lambda_{ci}=0} \quad \mathop{\stackrel{H_1}{\gtrless}}_{H_0} \quad  \log\lambda\\
&\Leftrightarrow& \lambda_{ci}\sum\limits_{i=1}^{N}\dfrac{d}{d\lambda_{ci}} \Bigl(-\frac{1}{2}\log (2\pi(\lambda_{ci}+\lambda_b+\sigma_w^2))\\
&&\qquad\qquad-\dfrac{(z_i-\lambda_{ci}-\lambda_b)^2}{2(\lambda_{ci}+\lambda_b+\sigma_w^2)}\Bigr)|_{\lambda_{ci}=0} \quad \mathop{\stackrel{H_1}{\gtrless}}_{H_0} \quad  \log\lambda\\
&\Leftrightarrow& \sum\limits_{i=1}^{N} (z_i-\lambda_b)+\sum\limits_{i=1}^{N} \dfrac{(z_i-\lambda_b)^2}{2(\lambda_b+\sigma_w^2)}\quad \mathop{\stackrel{H_1}{\gtrless}}_{H_0} \quad  (\lambda_b+\sigma_w^2)\log \lambda +\dfrac{N}{2}.
\end{eqnarray*}
\end{small}
\end{proof}

The resulting test statistic is independent of the unknown parameter $\lambda_{ci}$, and is the \textit{uniformly most powerful} (UMP) test for weak signals. 
\begin{figure*}[t]
\centering
\subfigure[]{
\includegraphics[%
  width=0.23\textwidth,clip=true]{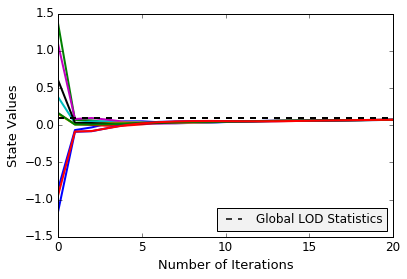}
\label{conv} }
\subfigure[]{
\includegraphics[%
  width=0.23\textwidth,clip=true]{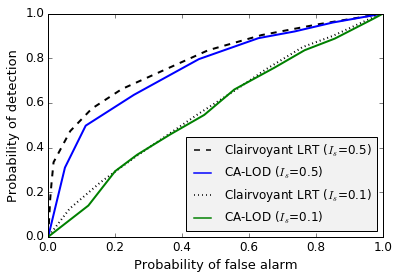}
\label{lod_p}}
\hfill
\subfigure[]{
\includegraphics[%
  width=0.23\textwidth,clip=true]{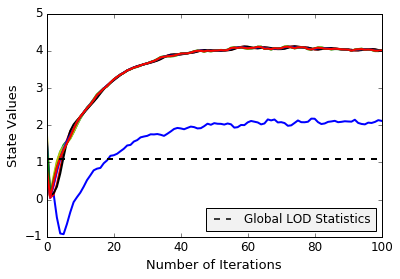}
\label{admm_conv} }
\subfigure[]{
\includegraphics[%
  width=0.23\textwidth,clip=true]{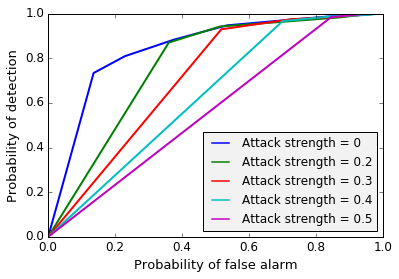}
\label{roc_b}}
\caption{\subref{conv} Convergence of state values of a network with $10$ nodes using ADMM based CA-LOD scheme. \subref{lod_p} Performance comparison of CA-LOD with clairvoyant LRT. \subref{admm_conv} Convergence of vanilla ADMM based CA-LOD in the presence of Byzantines. Blue curve represents Byzantine's state values. \subref{roc_b} Susceptibility of CA-LOD to Byzantine attack in terms of ROC.}
\label{perf1}
\vspace{-0.1in}
\end{figure*}
\subsection{Collaborative Autonomous Detection Using ADMM}
The LOD test statistic derived in the previous section is of the form below:
\begin{equation*}
\dfrac{1}{N}\sum \limits_{i=1}^N f(z_i)\quad \mathop{\stackrel{H_1}{\gtrless}}_{H_0} \quad  \dfrac{\gamma}{N}
\end{equation*}
where $f(z_i) = (z_i-\lambda_b)+\dfrac{(z_i-\lambda_b)^2}{2(\lambda_b+\sigma_w^2)}$. The LOD statistic is separable and the function $f(z_i)$ is strongly convex. Next, we show that the LOD statistic can be implemented in a distributed manner using ADMM. To apply ADMM, we first formulate a convex optimization problem
\begin{equation}
\label{eq1}
x^* = \arg\min_{\hat{x}} \sum\limits_{i=1}^N \dfrac{(\hat{x}-f(z_i))^2}{2}
\end{equation}
where the data average is the solution to a least-squares minimization problem. Next, we reformulate \eqref{eq1} in the ADMM amenable form as below
\begin{eqnarray}
&&\text{minimize}_{\{x_i\},\{y_{ij}\}} \quad \sum\limits_{i=1}^N \dfrac{(x_i-f(z_i))^2}{2}\\
&&\text{subject to} \qquad x_i=y_{ij}, x_j=y_{ij}, \forall(i,j)\in \mathcal{A}
\end{eqnarray}
where $\mathcal{A}$ is the adjacency matrix, $x_i$ is the local copy of the common optimization
variable $\hat{x}$ at node $i$ and $y_{ij}$ is an auxiliary variable imposing
the consensus constraint on neighboring nodes $i$ and $j$. In the matrix form, let us denote $F(\mathbf{x})=\frac{1}{2}\|\mathbf{x}-f(\mathbf{z})\|_2^2$, then, the optimization problem is 
\begin{small}
\begin{eqnarray}
&&\text{minimize}_{\mathbf{x,y}} \quad F(\mathbf{x})+G(\mathbf{y})\nonumber \\
&&\text{subject to} \qquad \mathbf{A}\mathbf{x} + \mathbf{B}\mathbf{y} = \mathbf{0}\label{admm-eq1}
\end{eqnarray}
\end{small}
where $G(\mathbf{y})=0$. Here $\mathbf{B} = [-\mathbf{I}_{|\mathcal{A}|};-\mathbf{I}_{|\mathcal{A}|}]$ and $\mathbf{A} = [\mathbf{A}_1;\mathbf{A}_2]$ with $\mathbf{A}_k\in \mathbb{R}^{2E\times N}$. If $(i,j)\in \mathcal{A}$ and $y_{ij}$ is the $q$th entry
of $\mathbf{y}$, then the $(q, i)$th entry of $\mathbf{A}_1$ and the $(q, j)$th entry of
$\mathbf{A}_2$ are $1$; otherwise the corresponding entries are $0$.
The augmented Lagrangian of~\eqref{admm-eq1} is given by
$$L_{\rho}(\mathbf{x},\mathbf{y},\mathbf{\lambda})=F(\mathbf{x})+\langle \mathbf{\lambda}, \mathbf{A}\mathbf{x} + \mathbf{B}\mathbf{y}\rangle+\dfrac{\rho}{2}\|\mathbf{A}\mathbf{x} + \mathbf{B}\mathbf{y}\|_2^2,$$
where $\mathbf{\lambda}=[\mathbf{\beta}_1;\mathbf{\beta}_2]$ with $\mathbf{\beta}_1,\mathbf{\beta}_2\in \mathbb{R}^{2E}$ is the Lagrange multiplier and $\rho$ is a positive algorithm parameter. The updates for ADMM are
\begin{small}
\begin{eqnarray}
&& \mathbf{x}\text{-update}: \nabla F(\mathbf{x}^{k+1})+\mathbf{A}^T\mathbf{\lambda}^k +\rho \mathbf{A}^T(\mathbf{A}\mathbf{x}^{k+1} + \mathbf{B}\mathbf{y}^{k}) = \mathbf{0}, \nonumber \\
&& \mathbf{y}\text{-update}: \mathbf{B}^T\mathbf{\lambda}^k +\rho \mathbf{B}^T(\mathbf{A}\mathbf{x}^{k+1} + \mathbf{B}\mathbf{y}^{k+1}) = \mathbf{0}, \nonumber\\
&& \mathbf{\lambda}\text{-update}: \mathbf{\lambda}^{k+1}-\mathbf{\lambda}^k -\rho (\mathbf{A}\mathbf{x}^{k+1} + \mathbf{B}\mathbf{y}^{k+1}) = \mathbf{0},\label{admm-update}
\end{eqnarray}
\end{small}
where $\nabla F(\mathbf{x}^{k+1}) = \mathbf{x}^{k+1}-f(\mathbf{z})$ is the gradient of $F(.)$ at $\mathbf{x}^{k+1}$. The global convergence of ADMM was established in~\cite{boyd_admm}. Since our objective function $F(\mathbf{x})$ is strongly convex
in $\mathbf{x}$, we obtain $x^*$ equal to the global test statistic as given in~\eqref{lodstat} as the unique solution. 

The updates in~\eqref{admm-update} can be further simplified to~\cite{shi2014linear},
\begin{small}
\begin{eqnarray}
&& x_i^{k+1} = \dfrac{1}{1+2\rho |\mathcal{N}_i|}\left(\rho |\mathcal{N}_i|x_i^{k} +\rho \sum_{j \in \mathcal{N}_i}x_j^{k}-\alpha_i^k+f(z_i)\right),\nonumber \\
&& \alpha_i^{k+1} = \alpha_i^k +\rho \left(|\mathcal{N}_i|x_i^{k+1}-\sum_{j \in \mathcal{N}_i}x_j^{k+1}\right)\label{admm-updates}
\end{eqnarray}
\end{small}
at node $i$ where $\mathcal{N}_i$ denotes the set of neighbors of node $i$. Note that, the updates in~\eqref{admm-updates} only depend on the data from the neighbors of the node $i$ and can be implemented in a fully autonomous manner. This implies that with these updates, each node can learn the global LOD test statistic only by local information exchanges. 

Next, to gain insight into the solution, we present illustrative examples that corroborate our
results. We consider a $10$ node network employing the ADMM updates as given in~\eqref{admm-updates} to determine the presence (or absence) of a radioactive source. Source and node locations and adjacency matrix were generated randomly in a region of interest of dimension $3.0\times 3.0$ units. The ADMM parameter $\rho$ was set to $1.0$. We assume a mean background radiation with count $\lambda_b= 0.5$ and measurement noise with $\sigma_w^2 = 0.5$. We further assume that  the prior probability of hypothesis is $P_0=P_1= 0.5$ and detection performance is empirically found by performing $1000$ Monte-Carlo runs. 

\subsubsection{Convergence Analysis}
To better understand the convergence properties of the proposed approach, we next present an instance of ADMM based CA-LOD in Fig.~\ref{conv}. We assume that each node starts with its local LOD statistic and collaborate with its neighbors to improve its performance. We plot the updated state values (LOD statistic) at each node as a function of information exchange iterations. 
Fig.~\ref{conv} shows the state values of each node as a function of the number of  iterations. We see that the state values converges to the global statistic within $20$ iterations using local interactions.


\subsubsection{Detection Performance Analysis}
Next, we analyze the detection performance of the proposed scheme.
In Fig.~\ref{lod_p}, we plot steady state receiver operating characteristic (ROC) curves for the proposed CA-LOD approach for different source intensities $\IntS$. We compare the performance of the proposed approach with clairvoyant LRT based approach which has knowledge of the true source location. For both $\IntS=0.1$ and $\IntS=0.5$, the proposed CA-LOD approach performs almost as good as the clairvoyant LRT based approach.

%
%

\begin{figure*}[t]
\centering
\subfigure[]{
\includegraphics[%
  width=0.23\textwidth,clip=true]{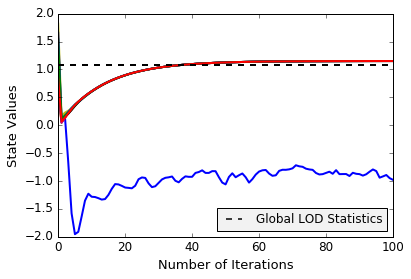}
\label{radmm_conv} }
\subfigure[]{
\includegraphics[%
  width=0.23\textwidth,clip=true]{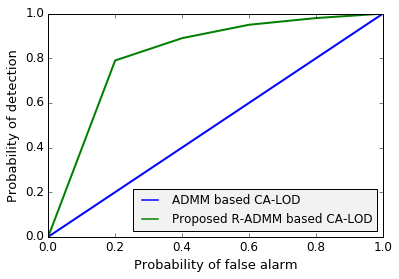}
\label{roc_comp}}
\subfigure[]{
\includegraphics[%
  width=0.23\textwidth,clip=true]{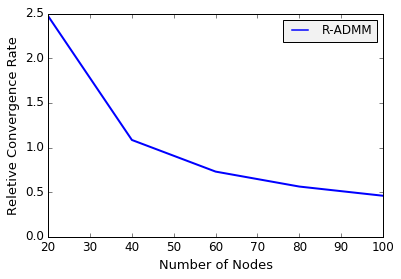}
\label{scaling} }
\subfigure[]{
\includegraphics[%
  width=0.23\textwidth,clip=true]{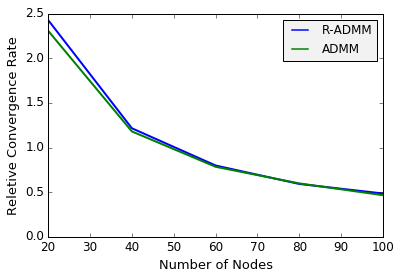}
\label{overhead}}
\caption{\subref{radmm_conv} Convergence of proposed algorithm based CA-LOD in the presence of Byzantine attack. Blue curve represents Byzantines state values. \subref{roc_comp} Detection performance of CA-LOD in the presence of Byzantine attacks. \subref{scaling} Scaling behavior of the proposed algorithm for bounded neighborhood size. \subref{overhead} Overhead comparison in the absence of Byzantine attacks.}
\label{perf2}
\vspace{-0.1in}
\end{figure*}

\section{Collaborative Autonomous Detection in the Presence of Byzantine Attacks}

\subsection{Byzantine Attack Model: Modus Operandi}
Note that, the ADMM update at node $i$ at iteration $k$ is a function of its neighbors' parameters $\{x_j^k\}_{j \in \mathcal{N}_i}$. Instead of broadcasting the true parameters $\{x_j^k\}$, some
nodes (referred to as Byzantines) can deviate from the prescribed strategies.    
More specifically, we assume that the Byzantine node $j$ falsifies its data at ADMM iteration $k$ as follows:
\begin{eqnarray*}
x_j^k &=& x_j^k + \delta_j^x
\end{eqnarray*}
where $\delta_j^x \sim \mathcal{N}(\mu_x,\,\sigma_x^{2})$. The strength of the attack is characterized by $(\mu_x,\,\sigma_x^{2})$.

\subsubsection{Performance Analysis of CA-LOD with Byzantines}
In this section, we study the susceptibility of CA-LOD in the presence of Byzantine attacks. 
We assume that there is only $1$ Byzantine in the network which is chosen randomly.  

In Fig.~\ref{admm_conv}, we plot the convergence of the ADMM algorithm with updates as given in~\eqref{admm-updates}. We assume the Byzantine's parameters to be
$\mu_x = 1.5$ and $\sigma_x^2 = 0.1$. It can be seen that the Byzantine attack can severely
degrade the convergence performance. More specifically, it can be seen from Fig.~\ref{admm_conv} that a single Byzantine can make the rest of the network converge to a state value which is significantly different from the global LOD statistic. 

Next, in Fig.~\ref{roc_b}, we plot the steady state ROC for different values of attack strength $\mu_x$ keeping $\sigma_x^2 $ fixed to $0.1$. Observe that, as the attack strength increases, the detection performance degrades severely and an adversary can make the steady state statistic (or data) non-informative. In other words, the optimal detection scheme at each node performs no better than a coin flip detector. 


%

\subsection{Robust Collaborative Autonomous Detection using Byzantine-Resilient ADMM}
Our approach draws inspiration from robust statistic for anomaly detection to make ADMM resilient to Byzantine attacks.
More specifically, we propose the following robust ADMM algorithm to tolerate at most $p$ Byzantines

\begin{small}
\begin{eqnarray}
&& x_i^{k+1} = \dfrac{1}{1+2\rho |\mathcal{N}_i|}\left(\rho |\mathcal{N}_i|x_i^{k} +\rho \Gamma_p(\{x_j^{k}\}_{j\in \mathcal{N}_i})-\alpha_i^k+f(z_i)\right),\nonumber \\
&& \alpha_i^{k+1} = \alpha_i^k +\rho \left(|\mathcal{N}_i|x_i^{k+1}-\Gamma_p(\{x_j^{k+1}\}_{j\in \mathcal{N}_i})\right)\label{radmm-updates}
\end{eqnarray}
\end{small}

\noindent where the sum over neighbors' data in~\eqref{admm-updates} has been replaced by a robust function $\Gamma_p(\{x_j^{k}\}_{j\in \mathcal{N}_i})$ which operates as follows: 

\noindent \textbf{Operation of} $\Gamma_p(.)$: \textit{First, sort the elements in $\mathcal{S}=\{x_j^{k}\}_{j\in \mathcal{N}_i}$ in a non-decreasing order (breaking ties arbitrarily), and replace the smallest $p$ values
and the largest $p$ values with mean of remaining $(|\mathcal{N}_i|-2p)$ values.\footnote{We assume that $|\mathcal{N}_i|>2p,\; \forall i$.} Next, return the sum of the elements in the new set.} 
%
%

Next, we analyze the performance of the proposed Byzantine-resilient autonomous detection scheme in the presence of Byzantine attacks. We assume $p=1$.

\subsubsection{Robustness Analysis}
%
%

In Fig.~\ref{radmm_conv}, we plot the convergence of the proposed R-ADMM algorithm with updates as given in~\eqref{radmm-updates}. We assume the Byzantine's parameters to be $\mu_x = 1.5$ and $\sigma_x^2 = 0.1$. It can be seen that, as opposed to Fig.~\ref{admm_conv}, the state values of the honest nodes converge close to the global LOD statistic despite the presence of Byzantine attack.

Next, in Fig.~\ref{roc_comp}, we compare the steady state ROC for CA-LOD of vanilla ADMM based approach with the R-ADMM based approach. We assume attack parameters to be $\mu_x = 2.5$ and $\sigma_x^2=0.1$. It can be seen that the R-ADMM based  Byzantine-resilient CA-LOD approach performs significantly better compare to the vanilla ADMM based approach, which breaks down in the the presence of the Byzantine attack.

\subsubsection{Scaling Analysis}
In Fig.~\ref{scaling}, we plot the convergence behavior of R-ADMM based CA-LOD as network grows larger. We consider a practical scenario where we fix the number of nodes (or neighbors) each node can talk to to be $10$. 
We plot relative convergence rates defined as $T^*/N$ where $T^*$ is the number of iterations needed to reach within $95$\% of the global LOD statistic. Note that, the convergence rate $T^*$ increases as number of nodes $N$ increases in the network, however, the relative convergence rate decreases. This implies that the proposed approach retains the excellent scaling properties of ADMM and is amenable for large scale networks.

In Fig.~\ref{overhead}, we compare the overhead caused by the R-ADMM based CA-LOD scheme. We consider the case where there is no Byzantine in the network and compare the performance of ADMM based CA-LOD and R-ADMM based CA-LOD in terms of relative convergence rate. It can be seen that the overhead caused by the R-ADMM based CA-LOD scheme is very small. In practice, this overhead is dominated by the sorting step in R-ADMM algorithm and is a constant for a bounded neighborhood. 
\vspace{-0.1in}
\section{Conclusion and Future Work}
In this paper, we proposed a decentralized locally optimum detection scheme for radioactive source detection. We also devised a robust version of the ADMM algorithm for Byzantine-resilient detection and demonstrated its robustness to data falsification attacks. 
There are still many interesting questions that remain to be explored in the future work such as analysis and extension of the problem with more realistic signal and communications models and collaborative Byzantine attacks. 

\bibliographystyle{IEEEtran}
\bibliography{references}

\begin{thebibliography}{10}
\providecommand{\url}[1]{#1}
\csname url@samestyle\endcsname
\providecommand{\newblock}{\relax}
\providecommand{\bibinfo}[2]{#2}
\providecommand{\BIBentrySTDinterwordspacing}{\spaceskip=0pt\relax}
\providecommand{\BIBentryALTinterwordstretchfactor}{4}
\providecommand{\BIBentryALTinterwordspacing}{\spaceskip=\fontdimen2\font plus
\BIBentryALTinterwordstretchfactor\fontdimen3\font minus
  \fontdimen4\font\relax}
\providecommand{\BIBforeignlanguage}[2]{{%
\expandafter\ifx\csname l@#1\endcsname\relax
\typeout{** WARNING: IEEEtran.bst: No hyphenation pattern has been}%
\typeout{** loaded for the language `#1'. Using the pattern for}%
\typeout{** the default language instead.}%
\else
\language=\csname l@#1\endcsname
\fi
#2}}
\providecommand{\BIBdecl}{\relax}
\BIBdecl

\bibitem{nd1}
S.~M. Brennan, A.~M. Mielke, and D.~C. Torney, ``Radioactive source detection
  by sensor networks,'' \emph{IEEE Transactions on Nuclear Science}, vol.~52,
  no.~3, pp. 813--819, 2005.

\bibitem{nd2}
D.~L. Stephens and A.~J. Peurrung, ``Detection of moving radioactive sources
  using sensor networks,'' \emph{IEEE Transactions on Nuclear Science},
  vol.~51, no.~5, pp. 2273--2278, 2004.

\bibitem{ashok_fusion}
A.~Sundaresan, P.~K. Varshney, and N.~S.~V. Rao, ``Distributed detection of a
  nuclear radioactive source using fusion of correlated decisions,'' in
  \emph{2007 10th International Conference on Information Fusion}, July 2007,
  pp. 1--7.

\bibitem{nd3}
C.~D. Pahlajani, I.~Poulakakis, and H.~G. Tanner, ``Networked decision making
  for poisson processes with applications to nuclear detection,'' \emph{IEEE
  Transactions on Automatic Control}, vol.~59, no.~1, pp. 193--198, Jan 2014.

\bibitem{paper4}
W.~Zhang, Z.~Wang, Y.~Guo, H.~Liu, Y.~Chen, and J.~Mitola, ``{Distributed
  Cooperative Spectrum Sensing Based on Weighted Average Consensus},'' in
  \emph{Global Telecommunications Conference (GLOBECOM 2011), 2011 IEEE}, Dec
  2011, pp. 1--6.

\bibitem{paper5}
F.~Pasqualetti, A.~Bicchi, and F.~Bullo, ``{Consensus Computation in Unreliable
  Networks: A System Theoretic Approach},'' \emph{Automatic Control, IEEE
  Transactions on}, vol.~57, no.~1, pp. 90--104, Jan 2012.

\bibitem{paper6}
G.~Xiong and S.~Kishore, ``{Consensus-based distributed detection algorithm in
  wireless ad hoc networks},'' in \emph{Signal Processing and Communication
  Systems, 2009. ICSPCS 2009. 3rd International Conference on}, Sept 2009, pp.
  1--6.

\bibitem{paper7}
S.~Aldosari and J.~Moura, ``{Distributed Detection in Sensor Networks:
  Connectivity Graph and Small World Networks},'' in \emph{Signals, Systems and
  Computers, 2005. Conference Record of the Thirty-Ninth Asilomar Conference
  on}, Oct 2005, pp. 230--234.

\bibitem{paper8}
S.~Kar, S.~Aldosari, and J.~Moura, ``{Topology for Distributed Inference on
  Graphs},'' \emph{Signal Processing, IEEE Transactions on}, vol.~56, no.~6,
  pp. 2609--2613, June 2008.

\bibitem{paper9}
F.~Yu, M.~Huang, and H.~Tang, ``{Biologically inspired consensus-based spectrum
  sensing in mobile Ad Hoc networks with cognitive radios},'' \emph{Network,
  IEEE}, vol.~24, no.~3, pp. 26--30, May 2010.

\bibitem{Lamport}
\BIBentryALTinterwordspacing
L.~Lamport, R.~Shostak, and M.~Pease, ``The byzantine generals problem,''
  \emph{ACM Trans. Program. Lang. Syst.}, vol.~4, no.~3, pp. 382--401, Jul.
  1982. [Online]. Available: \url{http://doi.acm.org/10.1145/357172.357176}
\BIBentrySTDinterwordspacing

\bibitem{frag}
A.~Fragkiadakis, E.~Tragos, and I.~Askoxylakis, ``A survey on security threats
  and detection techniques in cognitive radio networks,'' \emph{IEEE
  Communications Surveys Tutorials}, vol.~15, no.~1, pp. 428--445, 2013.

\bibitem{Rifa}
\BIBentryALTinterwordspacing
H.~Rif\`{a}-Pous, M.~J. Blasco, and C.~Garrigues, ``Review of robust
  cooperative spectrum sensing techniques for cognitive radio networks,''
  \emph{Wirel. Pers. Commun.}, vol.~67, no.~2, pp. 175--198, Nov. 2012.
  [Online]. Available: \url{http://dx.doi.org/10.1007/s11277-011-0372-x}
\BIBentrySTDinterwordspacing

\bibitem{Marano}
S.~Marano, V.~Matta, and L.~Tong, ``Distributed detection in the presence of
  byzantine attacks,'' \emph{IEEE Trans. Signal Process.}, vol.~57, no.~1, pp.
  16 --29, Jan. 2009.

\bibitem{Rawat}
A.~Rawat, P.~Anand, H.~Chen, and P.~Varshney, ``Collaborative spectrum sensing
  in the presence of byzantine attacks in cognitive radio networks,''
  \emph{IEEE Trans. Signal Process.}, vol.~59, no.~2, pp. 774 --786, Feb 2011.

\bibitem{bhavyaj}
B.~Kailkhura, S.~Brahma, Y.~S. Han, and P.~K. Varshney, ``{Distributed
  Detection in Tree Topologies With Byzantines},'' \emph{IEEE Trans. Signal
  Process.}, vol.~62, pp. 3208--3219, June 2014.

\bibitem{Kailkhura}
------, ``{Optimal Distributed Detection in the Presence of Byzantines},'' in
  \emph{Proc. The 38th International Conference on Acoustics, Speech, and
  Signal Processing (ICASSP 2013)}, Vancouver, Canada, May 2013.

\bibitem{aditya}
A.~Vempaty, K.~Agrawal, H.~Chen, and P.~K. Varshney, ``{Adaptive learning of
  Byzantines' behavior in cooperative spectrum sensing},'' in \emph{Proc. IEEE
  Wireless Comm. and Networking Conf. (WCNC)}, March 2011, pp. 1310 --1315.

\bibitem{bk_consensus}
B.~Kailkhura, S.~Brahma, and P.~K. Varshney, ``Data falsification attacks on
  consensus-based detection systems,'' \emph{IEEE Transactions on Signal and
  Information Processing over Networks}, vol.~3, no.~1, pp. 145--158, March
  2017.

\bibitem{tang}
H.~Tang, F.~Yu, M.~Huang, and Z.~Li, ``Distributed consensus-based security
  mechanisms in cognitive radio mobile ad hoc networks,'' \emph{Communications,
  IET}, vol.~6, no.~8, pp. 974--983, May 2012.

\bibitem{yu1}
F.~Yu, H.~Tang, M.~Huang, Z.~Li, and P.~Mason, ``{Defense against spectrum
  sensing data falsification attacks in mobile ad hoc networks with cognitive
  radios},'' in \emph{Military Communications Conference, 2009. MILCOM 2009.
  IEEE}, Oct 2009, pp. 1--7.

\bibitem{yu2}
F.~Yu, M.~Huang, and H.~Tang, ``{Biologically inspired consensus-based spectrum
  sensing in mobile Ad Hoc networks with cognitive radios},'' \emph{Network,
  IEEE}, vol.~24, no.~3, pp. 26--30, May 2010.

\bibitem{liu}
S.~Liu, H.~Zhu, S.~Li, X.~Li, C.~Chen, and X.~Guan, ``An adaptive
  deviation-tolerant secure scheme for distributed cooperative spectrum
  sensing,'' in \emph{Global Communications Conference (GLOBECOM), 2012 IEEE},
  Dec 2012, pp. 603--608.

\bibitem{yan}
Q.~Yan, M.~Li, T.~Jiang, W.~Lou, and Y.~Hou, ``Vulnerability and protection for
  distributed consensus-based spectrum sensing in cognitive radio networks,''
  in \emph{INFOCOM, 2012 Proceedings IEEE}, March 2012, pp. 900--908.

\bibitem{paper12}
Z.~Li, F.~Yu, and M.~Huang, ``{A Distributed Consensus-Based Cooperative
  Spectrum-Sensing Scheme in Cognitive Radios},'' \emph{Vehicular Technology,
  IEEE Transactions on}, vol.~59, no.~1, pp. 383--393, Jan 2010.

\bibitem{nedic}
A.~Nedic and A.~Ozdaglar, ``Distributed subgradient methods for multi-agent
  optimization,'' \emph{IEEE Transactions on Automatic Control}, vol.~54,
  no.~1, pp. 48--61, 2009.

\bibitem{duchi}
J.~C. Duchi, A.~Agarwal, and M.~J. Wainwright, ``Dual averaging for distributed
  optimization: Convergence analysis and network scaling,'' \emph{IEEE
  Transactions on Automatic Control}, vol.~57, no.~3, pp. 592--606, March 2012.

\bibitem{boyd_admm}
S.~Boyd, N.~Parikh, E.~Chu, B.~Peleato, and J.~Eckstein, ``Distributed
  optimization and statistical learning via the alternating direction method of
  multipliers,'' \emph{Foundations and Trends{\textregistered} in Machine
  Learning}, vol.~3, no.~1, pp. 1--122, 2011.

\bibitem{paper2}
R.~Olfati-Saber, J.~Fax, and R.~Murray, ``{Consensus and Cooperation in
  Networked Multi-Agent Systems},'' \emph{Proceedings of the IEEE}, vol.~95,
  no.~1, pp. 215--233, Jan 2007.

\bibitem{shi2014linear}
W.~Shi, Q.~Ling, K.~Yuan, G.~Wu, and W.~Yin, ``On the linear convergence of the
  admm in decentralized consensus optimization.'' \emph{IEEE Trans. Signal
  Processing}, vol.~62, no.~7, pp. 1750--1761, 2014.

\end{thebibliography}
\end{document}